\documentclass[english,12pt, authoryear]{article}
\usepackage[T1]{fontenc}
\usepackage[latin9]{inputenc}
\usepackage{geometry}
\usepackage{bm}
\usepackage{amsmath}
\usepackage{amsthm}
\usepackage{amssymb}
\usepackage{setspace}
\usepackage[authoryear]{natbib}
\makeatletter
\newenvironment{lyxlist}[1]
	{\begin{list}{}
		{\settowidth{\labelwidth}{#1}
		 \setlength{\leftmargin}{\labelwidth}
		 \addtolength{\leftmargin}{\labelsep}
		 }}
	{\end{list}}
\theoremstyle{plain}
\newtheorem{thm}{\protect\theoremname}
\theoremstyle{plain}
\newtheorem{lem}{\protect\lemmaname}
\theoremstyle{plain}
\newtheorem{prop}{\protect\propositionname}

\makeatother

\usepackage{babel}
\providecommand{\lemmaname}{Lemma}
\providecommand{\propositionname}{Proposition}
\providecommand{\theoremname}{Theorem}

\begin{document}
\title{Approximation results regarding the multiple-output Gaussian gated
mixture of linear experts model}
\author{Hien D. Nguyen$^{1}$\thanks{(Corresponding author email: h.nguyen5@latrobe.edu.au). $^{1}$Department
of Mathematics and Statistics, La Trobe University, Melbourne Victoria
3086, Australia. $^{2}$Laboratoire de Mathématiques Nicolas Oresme,
Université de Caen Normandie, 14000 Caen-Cedex, France. $^{3}$Université
de Grenoble Alpes, Inria, CNRS, Grenoble INP$^{\dagger}$, LJK, 38000
Grenoble, France. $^{\dagger}$Institute of Engineering Université
de Grenoble Alpes.}$\text{ }$, Faicel Chamroukhi$^{2}$, and Florence Forbes$^{3}$}
\maketitle

\subsection*{Abstract}

Mixture of experts (MoE) models are a class of artificial neural networks
that can be used for functional approximation and probabilistic modeling.
An important class of MoE models is the class of mixture of linear
experts (MoLE) models, where the expert functions map to real topological
output spaces. Recently, Gaussian gated MoLE models have become popular
in applied research. There are a number of powerful approximation
results regarding Gaussian gated MoLE models, when the output space
is univariate. These results guarantee the ability of Gaussian gated
MoLE mean functions to approximate arbitrary continuous functions,
and Gaussian gated MoLE models themselves to approximate arbitrary
conditional probability density functions. We utilize and extend upon
the univariate approximation results in order to prove a pair of useful
results for situations where the output spaces are multivariate. We
do this by proving a pair of lemmas regarding the combination of univariate
MoLE models, which are interesting in their own rights.

\textbf{Keywords: }artificial neural network; conditional model; Gaussian
distribution; mean function; multiple-output; multivariate analysis

\section{Introduction}

Mixture of experts (MoE) models are a class of probabilistic artificial
neural networks that were first introduced by \citet{Jacobs1991},
and further developed in \citet{Jordan1994} and \citet{Jordan1995}.
In the contemporary setting, MoE models have become highly popular
and successful in a range of applications including audio classification,
bioinformatics, climate prediction, face recognition, financial forecasting,
handwriting recognition, and text classification, among many others;
see \citet{Yuksel2012}, \citet{Masoudnia2014}, and \citet{Nguyen2018}
and the references therein.

Let $\mathbb{X}\subseteq\mathbb{R}^{p}$ and $\mathbb{Y}$ be input
and output spaces (the specific nature of the output space will be
discussed in the sequel), respectively, where $p\in\mathbb{N}$ (the
zero-exclusive natural numbers). Let $\bm{X}\in\mathbb{X}$ and $\bm{Y}\in\mathbb{Y}$
be observable random variables, where $\bm{X}$ may also be taken
to be non-stochastic (i.e. $\bm{X}=\bm{x}$ with probability one,
for some fixed $\bm{x}\in\mathbb{X}$). In addition to $\bm{X}$ and
$\bm{Y}$, define a third latent random variable $Z\in\left[n\right]=\left\{ 1,\dots,n\right\} $,
such that 
\begin{equation}
\mathbb{P}\left(Z=z|\bm{X}=\bm{x};\bm{\alpha}\right)=\text{Gate}_{z}\left(\bm{x};\bm{\alpha}\right)\text{,}\label{eq: Gates}
\end{equation}
where $\text{Gate}_{z}\left(\bm{x};\bm{\alpha}\right)$ are parametric
functions (known as gating functions), which depend on some vector
$\bm{\alpha}$ in a real space of fixed dimension. We call $n$ the
number of experts in the MoE. Here, the gating functions are required
to satisfy the conditions $\text{Gate}_{z}\left(\bm{x};\bm{\alpha}\right)>0$,
and $\sum_{z=1}^{n}\text{Gate}_{z}\left(\bm{x};\bm{\alpha}\right)=1$,
for each $z\in\left[n\right]$, $\bm{x}$, and $\bm{\alpha}$.

The probability density functions (PDFs) of $\bm{Y}$, given $\bm{X}=\bm{x}$
and $Z=z$, are referred to as expert functions, which are parametric
and can be written as
\begin{equation}
f\left(\bm{y}|\bm{x},z\right)=\text{Expert}_{z}\left(\bm{y};\bm{x},\bm{\beta}_{z}\right)\text{,}\label{eq: Experts}
\end{equation}
where $\bm{\beta}_{z}$ is a parameter vector in a real space of fixed
dimensionality, for each $z$. For brevity, we write $f\left(\bm{y}|\bm{x},z\right)=f\left(\bm{y}|\bm{X}=\bm{x},Z=z;\bm{\beta}_{z}\right)$.
We combine the gating functions (\ref{eq: Gates}) and expert functions
(\ref{eq: Experts}), via the law of total probability, to produce
the conditional PDF of $\bm{Y}$ given $\bm{X}=\bm{x}$:
\[
f\left(\bm{y}|\bm{x};\bm{\theta}\right)=\sum_{z=1}^{n}\text{Gate}_{z}\left(\bm{x};\bm{\alpha}\right)\text{Expert}_{z}\left(\bm{y};\bm{x},\bm{\beta}_{z}\right)\text{,}
\]
where, $\bm{\theta}$ is a vector that contains the elements of $\bm{\alpha}$
and $\bm{\beta}_{z}$ ($z\in\left[n\right]$). We refer to $f\left(\bm{y}|\bm{x};\bm{\theta}\right)=f\left(\bm{y}|\bm{X}=\bm{x};\bm{\theta}\right)$
as the MoE model.

Depending on the choices of gating and expert functions, numerous
classes of MoE models can be specified. For example, if $\mathbb{Y}$
is a binary or categorical output space, then one can consider a logistic
or multinomial logistic form (see, e.g. \citealp{Jordan1994,Chen1999}).
If $\mathbb{Y}\subset\mathbb{N}$, then one may follow \citet{Grun2008}
and utilize Poisson experts. When $\mathbb{Y}\subseteq\left(0,\infty\right)$
or $\mathbb{Y}\subseteq\left[0,1\right]$, the mixture of gamma or
beta experts are most appropriate (see, e.g. \citealp{Jiang1999a,Grun2012}).

In this article, we are only concerned with the case where $\mathbb{Y}\subseteq\mathbb{R}^{q}$
($q\in\mathbb{N}$), and when the mean of the expert functions are
linear in $\bm{x}$, so that
\begin{equation}
\mathbb{E}\left(\bm{y}|\bm{X}=\bm{x},Z=z\right)=\bm{a}_{z}+\mathbf{B}_{z}^{\top}\bm{x}=\left[\begin{array}{c}
a_{z1}+\bm{b}_{z1}^{\top}\bm{x}\\
\vdots\\
a_{zq}+\bm{b}_{zq}^{\top}\bm{x}
\end{array}\right]\text{,}\label{eq: Individual Mean Function}
\end{equation}
where we put the elements of $a_{zj}\in\mathbb{R}$ and $\bm{b}_{zj}\in\mathbb{R}^{p}$
($j\in\left[q\right]$) into $\bm{\beta}_{z}$, for each $z$. Here
$\left(\cdot\right)^{\top}$ is the transposition operator, $\bm{a}_{z}^{\top}=\left(a_{z,1},\dots,a_{z,q}\right)\in\mathbb{R}^{q}$,
and $\mathbf{B}_{z}\in\mathbb{R}^{p\times q}$ is a matrix with $j\text{th}$
column $\bm{b}_{z,j}$. Following the nomenclature of \citet{Nguyen2016},
we refer to MoE models with the characteristic above as mixture of
linear experts (MoLE) models. 

Define the $q\text{-dimensional}$ multivariate normal distribution
by its PDF
\[
\phi_{q}\left(\bm{y};\bm{\mu},\bm{\Sigma}\right)=\left|2\pi\bm{\Sigma}\right|^{-1/2}\exp\left[-\frac{1}{2}\left(\bm{y}-\bm{\mu}\right)^{\top}\bm{\Sigma}^{-1}\left(\bm{y}-\bm{\mu}\right)\right]\text{,}
\]
where $\bm{\mu}\in\mathbb{R}^{q}$ is a mean vector and $\bm{\Sigma}\in\mathbb{R}^{q\times q}$
is a symmetric positive-definite covariance matrix. The multivariate
normal linear experts were used to specify MoLE models in the foundational
works of \citet{Jacobs1991} and \citet{Jordan1994}. Alternative
MoLE models using Laplace, $\text{student-}t$, and skew $\text{student-}t$
linear experts have also been considered in \citet{Nguyen2016}, \citet{Chamroukhi2016},
and \citet{Chamroukhi2017}, respectively.

In the MoE literature, there are two dominant choices for gating functions.
The first, and by far the most popular, is the soft-max gate:
\begin{equation}
\text{Gate}_{z}\left(\bm{x};\bm{\alpha}\right)=\frac{\exp\left(c_{z}+\bm{d}_{z}^{\top}\bm{x}\right)}{\sum_{\zeta=1}^{n}\exp\left(c_{\zeta}+\bm{d}_{\zeta}^{\top}\bm{x}\right)}\text{,}\label{eq: Softmax gate}
\end{equation}
where $c_{z}\in\mathbb{R}$ and $\bm{d}_{z}\in\mathbb{R}^{p}$ ($z\in\left[n\right]$)
are put in the parameter vector $\bm{\alpha}$. This choice of gating
was originally considered in \citet{Jacobs1991}.

The second of the dominant gating functions is the Gaussian gating
function, or normalized-Gaussian radial basis gate (cf. \citealp{Wang1992}),
of the form
\begin{equation}
\text{Gate}_{z}\left(\bm{x};\bm{\alpha}\right)=\frac{\pi_{z}\phi_{p}\left(\bm{x};\bm{\mu}_{z},\bm{\Sigma}_{z}\right)}{\sum_{\zeta=1}^{n}\pi_{\zeta}\phi_{p}\left(\bm{x};\bm{\mu}_{\zeta},\bm{\Sigma}_{\zeta}\right)}\text{,}\label{eq: Gaussian gate}
\end{equation}
where $\pi_{z}>0$ and $\sum_{z=1}^{n}\pi_{z}=1$, and the unique
elements of $\pi_{z}$, $\bm{\mu}_{z}$, and $\bm{\Sigma}_{z}$ ($z\in\left[n\right]$)
are put in the parameter vector $\bm{\alpha}$. This gating choice
was originally considered, in the MoE context, by \citet{Xu1995},
although it had been used in the radial basis functions context by
\citet{Wang1992}. The Gaussian gating function has recently gained
some popularity in the literature. For example, \citet{Ingrassia2012}
used the Gaussian gating function within the framework of cluster-weighted
modeling, and \citet{Deleforge2015a} used the Gaussian gates within
the locally-linear mapping framework. Here, both cluster-weighted
models and locally-linear mappings are types of MoE models. The Gaussian
gates have also been used by \citet{Norets2014} and \citet{Norets2017}
for MoE modeling of priors in Bayesian nonparametric regression. Under
some restrictions, one can show that the class of soft-max gates is
a subset of the Gaussian gates (cf. \citealp[Cor. 5]{Ingrassia2012}).

A class of related gating functions to (\ref{eq: Gaussian gate})
are the $\text{student-}t$ gates. This type of gating has been explored
in \citet{Ingrassia2012}, \citet{Ingrassia2014}, and \citet{Perthame2018}.
Multivariate probit gates have also been considered in \citet{Geweke2007a}.

Given any particular choice of gating, we can write the MoLE mean
function as
\[
\bm{m}\left(\bm{x};\bm{\theta}\right)=\mathbb{E}\left(\bm{y}|\bm{X}=\bm{x}\right)=\sum_{z=1}^{n}\text{Gate}_{z}\left(\bm{x};\bm{\alpha}\right)\left[\bm{a}_{z}+\mathbf{B}_{z}^{\top}\bm{x}\right]\text{.}
\]
An important property of MoLE models is their richness of representation
capability. This representational richness has been characterized
in a number of ways via various theoretical results. In \citet{Zeevi1998}
and \citet{Jiang1999}, the single-output ($q=1$) soft-max gated
MoLE mean function was proved to be dense in an appropriate Sobolev
space, under assumptions on differentiability and measurability. We
define the notion of denseness in the manner of \citet[Ch. 22]{CheneyLight2000},
in the sequel. In \citet{Wang1992}, the single-output Gaussian gated
MoLE mean function was proved to be dense in the class of continuous
functions, using the Stone-Weierstrass theorem (cf. \citealp{Stone1948}).
Also via the Stone-Weierstrass theorem, \citet{Nguyen2016a} proved
that the single-output soft-max gated MoLE mean function is dense
in the class of continuous functions.

Distributional approximation theorems have also been obtained. For
example, \citet{Jiang1999a} proved that the class of single-output
soft-max gated MoLE models can approximate any conditional density
with mean function characterized via a ridge-type relationship with
the input vector (cf. \citealp{Pinkus2015}) to an arbitrary degree
of accuracy, with respect to the Hellinger distance and Kullback-Leibler
divergence (see \citealp[Ch. 3]{Pollard2002}). Replacing the linear
mean functions (\ref{eq: Individual Mean Function}) by polynomials,
\citet{MendesJiang2012} obtained an approximation result regarding
conditional densities with Sobolev class mean functions, instead of
ridge-type mean functions. 

We note that the results of \citet{Jiang1999a}, \citet{Jiang1999},
and \citet{MendesJiang2012} are more general than what has been discussed
here. That is, the results from the aforementioned papers extend to
various generalized linear MoE models, and are not restricted to the
MoLE context.

In a similar manner to \citet{Jiang1999a} and \citet{MendesJiang2012},
\citet{Norets2010} and \citet{Pelenis2014} showed that the single-output
soft-max gated MoLE models can approximate any conditional density,
regardless of mean function (under some regularity conditions), to
an arbitrary degree of accuracy, with respect to a Kullback-Leibler
type divergence. Extending upon the results of \citet{Norets2010}
and \citet{Pelenis2014}, \citet{Norets2014} proved that the same
approximation result holds for Gaussian gated MoLE models.

In recent years, numerous articles have described practical applications
of multi-output MoLE models ($q>1$; MO). For example, \citet{Chamroukhi2013a}
utilized such models for time series segmentation of human activity
data. An application of MO-MoLE models to the analyze genomics data
appears in \citet{Montuelle2014}. Such models have also been used
in image reconstruction and spectroscopic remote sensing applications
\citep{Deleforge2015}, as well as in sound source separation applications
\citep{Deleforge2015a}. Time series applications of MO-MoLE models
have been considered by \citet{Prado2006} and \citet{Kalliovirta2016}.

Unfortunately, the single-output approximation theorems that have
been previously cited no longer apply directly to MO-MoLE models.
This is because there is no currently available results that allow
for the pooling of marginal univariate effects, to best of our knowledge.
That is, one cannot simply assume that the individual modeling of
each output variable as an MoLE results in an MO-MoLE model when viewed
across all of the variables, simultaneously, to the best of our knowledge,
this current work is the first article to establish such a result
via Lemmas \ref{lem: closure under add} and \ref{lem: closure under mult}.

In this paper, we utilize the previous results of \citet{Wang1992}
and \citet{Norets2014} in order to state useful approximation theorems
to justify the use of MO-MoLE models for the analysis of functionally
complex data and those data that arise from complex distributions.
The approximation theorems regarding MO-MoLE models are presented
as Theorems \ref{thm: MO Norets} and \ref{thm: Denseness of Mean}.

Theorem \ref{thm: MO Norets} states that we can arbitrarily well
approximate the marginal conditional densities of any multivariate
regression data generating process (DGP) in the conditional Kullback-Leibler
divergence, provided we utilize an MO-MoLE model with a sufficiently
large but finite number of experts. Similarly, Theorem \ref{thm: Denseness of Mean}
states that we can utilize the mean function of an MO-MoLE model to
arbitrarily well estimate all output variables of a continuous multivariate
function over a compact support, simultaneously, given a sufficiently
large but finite number of experts in our model. Recently there has
been some interest in the use of deep variants of MO-MoLE models for
multivariate density estimation and functional approximation (see,
e.g., \citealp{Shazeer:2017aa}, \citealp{Fu:2018aa}, and \citealp{Zhao:2018aa}).
Our results provide empirical justification for the empirical effectiveness
in modeling complex multivariate data of these deep variants and the
already noted shallow counterparts that have been cited earlier.

In order to prove Theorems \ref{thm: MO Norets} and \ref{thm: Denseness of Mean},
we also proved a pair of technical lemmas regarding the combination
of univariate MoLE models. These lemmas are interesting in their own
rights, and are presented as Lemmas \ref{lem: closure under add}
and \ref{lem: closure under mult}.

To the best of our knowledge, the approximation capabilities of the
MO-MoLE models have not considered in previous articles on the topic.
This is because past works have primarily focused on the derivation
of estimation algorithms for MO-MoLE algorithms and the probabilistic
properties of the estimators of such models under various DGPs. The
assumption that MO-MoLE models would provide good approximations extrapolated
from works regarding the related class of finite mixture models (cf.
\citealp[Sec. 33.1]{DasGupta2008} and \citealp{Norets2012}). Our
results are the first available theorems that explain the empirical
effectiveness of MO-MoLE models in practice.

The rest of the paper is organized as follows. The univariate results
of \citet{Wang1992} and \citet{Norets2014} are presented in Section
2. The main results of the paper are stated in Section 3. Proofs of
the main theorems are provided in Section 4. Discussions and conclusions
are presented in Section 5. Supporting results are reported in the
Appendix.

\section{Preliminary results}

The approximation result of \citet{Norets2014} requires the following
setup. Suppose that we observe the data pair $\left(\bm{X},Y\right)\in\mathbb{X}\times\mathbb{Y}$,
where $\mathbb{Y}\subset\mathbb{R}$, is generated from a DGP that
can be characterized by a marginal PDF $g_{\bm{X}}\left(\bm{x}\right)$
and conditional PDF $g_{Y|\bm{X}}\left(y|\bm{x}\right)$. Let $G$
denote the joint probability measure that is implied by the joint
PDF $g_{Y|\bm{X}}\left(y|\bm{x}\right)g_{\bm{X}}\left(\bm{x}\right)$. 

Make the assumptions that: 
\begin{lyxlist}{00.00.0000}
\item [{{[}A1{]}}] $g_{Y|\bm{X}}\left(y|\bm{x}\right)$ is a continuous
function in both $\bm{x}\in\mathbb{X}$ and $y\in\mathbb{Y}$, almost
surely with respect to $G$, and 
\item [{{[}A2{]}}] there exists some $\rho>0$ such that
\[
\int_{\mathbb{X}\times\mathbb{Y}}\log\frac{g_{Y|\bm{X}}\left(y|\bm{x}\right)}{\inf_{\left\{ \left(\bm{s},t\right):\left\Vert \bm{x}-\bm{s}\right\Vert \le\rho,\left\Vert y-t\right\Vert \le\rho\right\} }g_{Y|\bm{X}}\left(t|\bm{s}\right)}\text{d}G\left(\bm{x},y\right)<\infty\text{,}
\]
where $\left\Vert \cdot\right\Vert $ is the Euclidean norm. 
\end{lyxlist}
As stated by \citet{Norets2014}, condition {[}A2{]} is a technical
requirement that the log relative changes in $g_{Y|\bm{X}}\left(y|\bm{x}\right)$
are finite, on average, and that $g_{Y|\bm{X}}\left(y|\bm{x}\right)$
is positive for all pairs of $\bm{x}$ and $y$. 

Write the class of MO-MoLE models with Gaussian gates and Gaussian
linear experts over $\mathbb{X}$ as
\[
\mathcal{L}_{q}\left(\mathbb{X}\right)=\left\{ f:f\left(\bm{y}|\bm{x};\bm{\theta}\right)=\sum_{z=1}^{n}\frac{\pi_{z}\phi_{p}\left(\bm{x};\bm{\mu}_{z},\bm{\Sigma}_{z}\right)\phi_{q}\left(\bm{y};\bm{a}_{z}+\mathbf{B}_{z}^{\top}\bm{x},\mathbf{C}_{z}\right)}{\sum_{\zeta=1}^{n}\pi_{\zeta}\phi_{p}\left(\bm{x};\bm{\mu}_{\zeta},\bm{\Sigma}_{\zeta}\right)},\text{ }n\in\mathbb{N}\right\} \text{.}
\]
Here, $\mathbf{C}_{z}\in\mathbb{R}^{q\times q}$ is a symmetric positive-definite
covariance matrix, for each $z\in\left[n\right]$. Furthermore, define
the subclass $\mathcal{L}_{q}^{*}\left(\mathbb{X}\right)\subset\mathcal{L}_{q}\left(\mathbb{X}\right)$,
where
\[
\mathcal{L}_{q}^{*}\left(\mathbb{X}\right)=\left\{ f:f\left(\bm{y}|\bm{x};\bm{\theta}\right)=\sum_{z=1}^{n}\frac{\pi_{z}\phi_{p}\left(\bm{x};\bm{\mu}_{z},\bm{\Sigma}_{z}\right)\phi_{q}\left(\bm{y};\bm{a}_{z},\mathbf{C}_{z}\right)}{\sum_{\zeta=1}^{n}\pi_{\zeta}\phi_{p}\left(\bm{x};\bm{\mu}_{\zeta},\bm{\Sigma}_{\zeta}\right)},\text{ }n\in\mathbb{N}\right\} \text{.}
\]
The following result is a direct consequence of \citet[Thm. 3.1]{Norets2014}.
\begin{thm}
\label{thm: Univariate case Norets}Let $\mathbb{X}$ be compact and
$\mathbb{Y}\subset\mathbb{R}$. If the data pair $\left(\bm{X},Y\right)$
arises from a DGP that is characterized by the joint probability measure
$G$, and if $g_{Y|\bm{X}}$ is a conditional PDF that satisfies {[}A1{]}
and {[}A2{]}, then for every $\epsilon>0$, there exist $n$ and $\bm{\theta}$
that characterize an MoLE model $f\in\mathcal{L}_{1}^{*}\left(\mathbb{X}\right)$,
such that
\[
\int_{\mathbb{X}\times\mathbb{Y}}\log\frac{g_{Y|\bm{X}}\left(y|\bm{x}\right)}{f\left(y|\bm{x};\bm{\theta}\right)}\text{d}G\left(\bm{x},y\right)<\epsilon\text{.}
\]
\end{thm}
We now consider the approximation theorem of \citet{Wang1992}. Let
$\mathcal{C}\left(\mathbb{X}\right)$ denote the class of all continuous
functions with support $\mathbb{X}\subset\mathbb{R}^{p}$. For a pair
of single-output functions $u$ and $v$ on $\mathbb{X}$, we can
define the uniform distance between $u$ and $v$ as $\text{d}_{\infty}\left(u,v\right)=\left\Vert u-v\right\Vert _{\infty}$,
where $\left\Vert u\left(\bm{x}\right)\right\Vert _{\infty}=\sup_{\bm{x}\in\mathbb{X}}\left|u\left(\bm{x}\right)\right|$
is the uniform norm over the support $\mathbb{X}$.

The following definition is taken from \citet[Ch. 22]{CheneyLight2000}.
Suppose that $\mathcal{U}\left(\mathbb{X}\right)$ and $\mathcal{V}\left(\mathbb{X}\right)$
are two classes of functions on $\mathbb{X}$. If $\mathcal{U}$ and
$\mathcal{V}$ are normed vector spaces (with respect to an appropriate
norm), then we say that $\mathcal{U}$ is dense in $\mathcal{V}$,
if the closure of $\mathcal{U}$ is $\mathcal{V}$. That is, we say
that $\mathcal{U}$ is dense in $\mathcal{V}$ with respect to the
uniform norm, if for each $v\in\mathcal{V}$ and $\epsilon>0$, there
exists a linear combination $u\in\mathcal{U}$, such that $\text{d}_{\infty}\left(u,v\right)<\epsilon$.

For $\mathbb{Y}=\mathbb{R}^{q}$, denote the class of Gaussian gated
MoLE mean functions over the support $\mathbb{X}$ by
\[
\mathcal{M}_{q}\left(\mathbb{X}\right)=\left\{ \bm{m}:\bm{m}\left(\bm{x};\bm{\theta}\right)=\sum_{z=1}^{n}\frac{\pi_{z}\phi_{p}\left(\bm{x};\bm{\mu}_{z},\bm{\Sigma}_{z}\right)}{\sum_{\zeta=1}^{n}\pi_{\zeta}\phi_{p}\left(\bm{x};\bm{\mu}_{\zeta},\bm{\Sigma}_{\zeta}\right)}\left[\bm{a}_{z}+\mathbf{B}_{z}^{\top}\bm{x}\right]\text{, }\bm{x}\in\mathbb{X}\text{, }n\in\mathbb{N}\right\} \text{.}
\]
Further define the subclass $\mathcal{M}_{q}^{*}\left(\mathbb{X}\right)\subset\mathcal{M}_{q}\left(\mathbb{X}\right)$,
where
\[
\mathcal{M}_{q}^{*}\left(\mathbb{X}\right)=\left\{ \bm{m}\in\mathcal{M}_{q}\left(\mathbb{X}\right):\mathbf{B}_{z}=\mathbf{0}\text{, for each }z\in\left[n\right]\right\} \text{,}
\]
and\textbf{ }$\mathbf{0}$ is a matrix containing only zeros of appropriate
dimensionality. In the $q=1$ case, the following result was proved
by \citet{Wang1992}, using the Stone-Weierstrass theorem.
\begin{thm}
\label{thm: Wang}If $\mathbb{X}\subset\mathbb{R}^{p}$ is a compact
set, then the set $\mathcal{M}_{1}^{*}\left(\mathbb{X}\right)$ is
dense in $\mathcal{C}\left(\mathbb{X}\right)$, with respect to the
uniform norm. Subsequently, since $\mathcal{M}_{1}^{*}\left(\mathbb{X}\right)\subset\mathcal{M}_{1}\left(\mathbb{X}\right)$,
it follows that $\mathcal{M}_{1}\left(\mathbb{X}\right)$ is also
dense in $\mathcal{C}\left(\mathbb{X}\right)$, with respect to the
uniform norm.
\end{thm}

\section{Main results}

Extending from the work of \citet{Norets2014}, we now consider the
approximation capabilities of MO-MoLE models. To do so, we require
the following definitions. 

Let $\mathbb{Y}=\prod_{j=1}^{q}\mathbb{Y}_{j}$, such that $\mathbb{Y}_{j}\subset\mathbb{R}$.
Suppose that the data pair $\left(\bm{X},\bm{Y}\right)\in\mathbb{X}\times\mathbb{Y}$
is generated from a DGP that can be characterized by a marginal PDF
$g_{\bm{X}}\left(\bm{x}\right)$ and that admits the univariate conditional
PDFs $g_{Y_{j}|\bm{X}}\left(y_{j}|\bm{x}\right)$, for each $j\in\left[q\right]$,
where $\bm{Y}^{\top}=\left(Y_{1},\dots,Y_{q}\right)$, and subsequently
$\bm{y}^{\top}=\left(y_{1},\dots,y_{q}\right)$. Let the probability
measure that is implied by the PDF $g_{Y_{j}|\bm{X}}\left(y_{j}|\bm{x}\right)g_{\bm{X}}\left(\bm{x}\right)$
be written as $G_{j}$, for each $j\in\left[q\right]$.

Make Assumptions {[}A1{]} and {[}A2{]} regarding each of the conditional
PDFs $g_{Y_{j}|\bm{X}}\left(y_{j}|\bm{x}\right)$. That is, assume
that:
\begin{lyxlist}{00.00.0000}
\item [{{[}B1{]}}] for each $j\in\left[q\right]$, $g_{Y_{j}|\bm{X}}\left(y_{j}|\bm{x}\right)$
is a continuous function in both $\bm{x}\in\mathbb{X}$ and $y_{j}\in\mathbb{Y}_{j}$,
almost surely with respect to $G_{j}$, and
\item [{{[}B2{]}}] for each $j\in\left[q\right]$, there exists some $\rho_{j}>0$
such that
\[
\int_{\mathbb{X}\times\mathbb{Y}_{j}}\log\frac{g_{Y_{j}|\bm{X}}\left(y|\bm{x}\right)}{\inf_{\left\{ \left(\bm{s},t\right):\left\Vert \bm{x}-\bm{s}\right\Vert \le\rho,\left\Vert y-t\right\Vert \le\rho\right\} }g_{Y_{j}|\bm{X}}\left(t|\bm{s}\right)}\text{d}G\left(\bm{x},y_{j}\right)<\infty\text{.}
\]
\end{lyxlist}
Using Theorem \ref{thm: Univariate case Norets}, we obtain the following
generalization regarding MO-MoLE models from the class $\mathcal{L}_{q}^{*}\left(\mathbb{X}\right)$,
and subsequently, the class $\mathcal{L}_{q}\left(\mathbb{X}\right)$.
The proof appears in Section 4.
\begin{thm}
\label{thm: MO Norets}Let $\mathbb{X}$ be compact and $\mathbb{Y}=\prod_{j=1}^{q}\mathbb{Y}_{j}$,
where $\mathbb{Y}_{j}\subset\mathbb{R}$. Assume that the DGP of $\left(\bm{X},\bm{Y}\right)$
is compatible with each of the joint probability measures $G_{j}$
($j\in\left[q\right]$). If the conditional PDFs $g_{Y_{j}|\bm{X}}$
($j\in\left[q\right]$) are such that Assumptions {[}B1{]} and {[}B2{]}
are satisfied, then there exist $n$ and $\bm{\theta}$ that characterize
an MoLE model $f\in\mathcal{L}_{q}\left(\mathbb{X}\right)$, such
that for some $\epsilon>0$,
\[
\int_{\mathbb{X}\times\mathbb{Y}_{j}}\log\frac{g_{Y_{j}|\bm{X}}\left(y_{j}|\bm{x}\right)}{f\left(y_{j}|\bm{x};\bm{\theta}\right)}\text{d}G_{j}\left(\bm{x},y_{j}\right)<\epsilon
\]
is satisfied simultaneously for all $j\in\left[q\right]$.
\end{thm}
We now extend upon the result of \citet{Wang1992} in order to state
a theorem regarding the approximation capabilities of MO-MoLE mean
functions. Define the space of MO continuous functions over $\mathbb{X}$
as
\[
\mathcal{C}_{q}\left(\mathbb{X}\right)=\left\{ \bm{m}^{\top}\left(\bm{x}\right)=\left(m_{1}\left(\bm{x}\right),\dots,m_{q}\left(\bm{x}\right)\right):m_{j}\in\mathcal{C}\left(\mathbb{X}\right),j\in\left[q\right]\right\} \text{.}
\]
We wish to determine the relationship between the class $\mathcal{M}_{q}\left(\mathbb{X}\right)$
and $\mathcal{C}_{q}\left(\mathbb{X}\right)$, for $q>1$. In order
to state such a relationship, we require an appropriate distance function.
Following the approach of \citet{Chiou2014}, we utilize summation
to induce a multivariate norm and distance function as follows. Let
$\bm{u}^{\top}=\left(u_{1},\dots,u_{q}\right)$ and $\bm{v}^{\top}=\left(v_{1},\dots,v_{q}\right)$
be a pair of MO functions on $\mathbb{X}$. Denote the induced distance
between $\bm{u}$ and $\bm{v}$ by $\text{d}_{q,\infty}\left(\bm{u},\bm{v}\right)=\left\Vert \bm{u}-\bm{v}\right\Vert _{q,\infty}$,
where $\left\Vert \bm{u}\left(\bm{x}\right)\right\Vert _{q,\infty}=\sum_{j=1}^{q}\left\Vert u_{j}\left(\bm{x}\right)\right\Vert _{\infty}$. 

We prove that the operator $\left\Vert \cdot\right\Vert _{q,\infty}$
satisfies the definition of a norm in the Appendix. Our following
result generalizes Theorem \ref{thm: Wang}. The proof appears in
Section 4.
\begin{thm}
\label{thm: Denseness of Mean}If $\mathbb{X}\subset\mathbb{R}^{p}$
is a compact set and $q\in\mathbb{N}$, then the sets of MO-MoLE mean
functions $\mathcal{M}_{q}^{*}\left(\mathbb{X}\right)$ and $\mathcal{M}_{q}\left(\mathbb{X}\right)$
are dense in $\mathcal{C}_{q}\left(\mathbb{X}\right)$, with respect
to the induced norm. 
\end{thm}
We note that both Theorems \ref{thm: Denseness of Mean} and \ref{thm: MO Norets}
require that the gating functions are of the Gaussian form, given
by (\ref{eq: Gaussian gate}). We note that \citet[Thm. 1]{Nguyen2016a}
provides a version of Theorem \ref{thm: Wang} that utilizes the soft-max
gating function instead of the Gaussian gating function, under the
same compactness assumption on $\mathbb{X}$. Similarly, \citet[Thm. 1]{Pelenis2014}
provides a substitute for Theorem \ref{thm: Univariate case Norets},
under almost identical assumptions, for soft-max gated MoLEs with
Gaussian linear experts. An additional assumption that $\int_{\mathbb{Y}}y^{2}g_{Y|\bm{X}}\left(y|\bm{x}\right)\text{d}\bm{x}<\infty$
for all $\bm{x}\in\mathbb{X}$ is required, in order to apply the
result of \citet{Pelenis2014}. Thus, one can largely replace the
Gaussian gating functions in Theorems \ref{thm: Denseness of Mean}
and \ref{thm: MO Norets} by the soft-max gating functions of form
(\ref{eq: Softmax gate}), and still obtain the conclusions of the
two results.

Theorems \ref{thm: MO Norets} and \ref{thm: Denseness of Mean} are
directly applicable to the MO-MoLE models that are considered in \citet{Prado2006},
\citet{Chamroukhi2013a}, \citet{Montuelle2014}, \citep{Deleforge2015},
\citep{Deleforge2015a}, and \citet{Kalliovirta2016}. For example,
the MO-MoLE models of \citet{Chamroukhi2013a} and \citealp{Deleforge2015}
take the forms:
\[
f\left(\bm{y}|\bm{x};\bm{\theta}\right)=\sum_{z=1}^{n}\frac{\exp\left(c_{z}+\bm{d}_{z}^{\top}\bm{x}\right)}{\sum_{\zeta=1}^{n}\exp\left(c_{\zeta}+\bm{d}_{\zeta}^{\top}\bm{x}\right)}\phi_{q}\left(\bm{y};\bm{a}_{z}+\mathbf{B}_{z}^{\top}\bm{x},\bm{\Omega}_{z}\right)
\]
and

\[
f\left(\bm{y}|\bm{x};\bm{\theta}\right)=\sum_{z=1}^{n}\frac{\pi_{z}\phi_{p}\left(\bm{x};\bm{\mu}_{z},\bm{\Sigma}_{z}\right)}{\sum_{\zeta=1}^{n}\pi_{\zeta}\phi_{p}\left(\bm{x};\bm{\mu}_{\zeta},\bm{\Sigma}_{\zeta}\right)}\phi_{q}\left(\bm{y};\bm{a}_{z}+\mathbf{B}_{z}^{\top}\bm{x},\bm{\Omega}_{z}\right)\text{,}
\]
where $\bm{\Omega}_{z}$ is a positive definite and symmetric matrix,
for each $z\in\left[n\right]$. Thus both MO-MoLE models satisfy the
assumptions of \ref{thm: MO Norets} and \ref{thm: Denseness of Mean}.
We can therefore conclude that with sufficiently many experts $n$,
both models are able to arbitrarily well approximate mean functions
and conditional marginal density functions of the underlying DGPs.
This therefore explains why these models, and the other cited MO-MoLE
models are able to well approximate their target functions in the
respective articles.

The distributional approximation and denseness results provide some
theoretical justification for the flexibility and goodness-of-fit
of such models in the simulation studies and applications that are
presented in the listed references. Furthermore, we note that the
results are directly applicable to any class of MoLE models with gating
functions that includes Gaussian as a subclass. For example, it is
hypothetically possible to construct a family of skew normal gated
MoLE models using the skew normal distributions of \citet{Azzalini1996},
where the skew normal density function replaces the Gaussian density
function in (\ref{eq: Gaussian gate}). Since the skew normal distribution
includes the Gaussian distribution as a special case, the results
of our theorems would immediately apply to such a construction.

\section{Proofs of main results}

The following lemmas streamline the proofs of Theorems \ref{thm: Denseness of Mean}
and \ref{thm: MO Norets}. The first lemma is well known and characterizes
the functional form of the product of two Gaussian PDFs. A proof of
the lemma can be found in \citet{Bromiley2014}. The proofs of Lemmas
\ref{lem: closure under add} and \ref{lem: closure under mult} appear
in the Appendix.
\begin{lem}
\label{lem: product}If $\bm{\mu}_{1},\bm{\mu}_{2}\in\mathbb{R}^{p}$
and $\bm{\Sigma}_{1},\bm{\Sigma}_{2}\in\mathbb{R}^{p\times p}$ are
symmetric positive-definite covariance matrices, then
\[
\phi_{p}\left(\bm{x};\bm{\mu}_{1},\bm{\Sigma}_{1}\right)\phi_{p}\left(\bm{x};\bm{\mu}_{2},\bm{\Sigma}_{2}\right)=c\phi_{p}\left(\bm{x};\bm{\mu}_{12},\bm{\Sigma}_{12}\right)\text{,}
\]
where $c>0$, $\bm{\Sigma}_{12}^{-1}=\bm{\Sigma}_{1}^{-1}+\bm{\Sigma}_{2}^{-1}$,
and $\bm{\mu}_{12}=\bm{\Sigma}_{12}\left(\bm{\Sigma}_{1}^{-1}\bm{\mu}_{1}+\bm{\Sigma}_{2}^{-1}\bm{\mu}_{2}\right)$.
\end{lem}
\begin{lem}
\label{lem: closure under add}If $\bm{m}_{\left[1\right]},\bm{m}_{\left[2\right]}\in\mathcal{M}_{q}^{*}\left(\mathbb{X}\right)$,
for some $\mathbb{X}$, then $\bm{m}_{\left[12\right]}\in\mathcal{M}_{q}^{*}\left(\mathbb{X}\right)$,
where $\bm{m}_{\left[12\right]}=\bm{m}_{\left[1\right]}+\bm{m}_{\left[2\right]}$.
\end{lem}
\begin{lem}
\label{lem: closure under mult} If $f_{\left[1\right]}\in\mathcal{L}_{q}^{*}\left(\mathbb{X}\right)$
and $f_{\left[2\right]}\in\mathcal{L}_{r}^{*}\left(\mathbb{X}\right)$,
for some $\mathbb{X}$ ($q,r\in\mathbb{N}$), then $f_{\left[12\right]}\in\mathcal{L}_{q+r}^{*}\left(\mathbb{X}\right)$,
where $f_{\left[12\right]}=f_{\left[1\right]}f_{\left[2\right]}$.
\end{lem}

\subsection{Proof of Theorem \ref{thm: MO Norets}}

By Theorem \ref{thm: Univariate case Norets}, under Assumptions {[}B1{]}
and {[}B2{]}, for each $j\in\left[q\right]$ and $\epsilon>0$, there
exists an $n_{j}$ and $\bm{\theta}_{j}$ that specifies a function

\[
f\left(y_{j}|\bm{x};\bm{\theta}_{j}\right)=\sum_{z=1}^{n_{j}}\frac{\pi_{jz}\phi_{p}\left(\bm{x};\bm{\mu}_{jz},\bm{\Sigma}_{jz}\right)\phi_{1}\left(y_{j};a_{jz},\sigma_{jz}^{2}\right)}{\sum_{\zeta=1}^{n_{j}}\pi_{j\zeta}\phi_{p}\left(\bm{x};\bm{\mu}_{j\zeta},\bm{\Sigma}_{j\zeta}\right)}\text{,}
\]
in $\mathcal{L}_{1}^{*}\left(\mathbb{X}\right)$, such that 
\[
\int_{\mathbb{X}\times\mathbb{Y}_{j}}\log\frac{g_{Y_{j}|\bm{X}}\left(y_{j}|\bm{x}\right)}{f\left(y_{j}|\bm{x};\bm{\theta}\right)}\text{d}G\left(\bm{x},y_{j}\right)<\epsilon
\]
is satisfied.

We complete the proof constructively. That is, we can show that the
product of the marginal PDFs $f\left(y_{j}|\bm{x};\bm{\theta}_{j}\right)$
yields a joint PDF $f\left(\bm{y}|\bm{x};\bm{\theta}\right)$, which
is in the class $\mathcal{L}_{q}^{*}\left(\mathbb{X}\right)$. This
is achieved via repeated applications of Lemma \ref{lem: closure under mult}.
We obtain the desired conclusion by noting that $\mathcal{L}_{q}^{*}\left(\mathbb{X}\right)\subset\mathcal{L}_{q}\left(\mathbb{X}\right)$.

\subsection{Proof of Theorem \ref{thm: Denseness of Mean}}

Let $\mathbb{X}$ be a compact set. Define $\mathbf{e}_{j}$ to be
a column vector with 1 in the $j\text{th}$ position and 0, elsewhere.
Let $\bm{u}^{\top}\left(\bm{x}\right)=\left(u_{1}\left(\bm{x}\right),\dots,u_{q}\left(\bm{x}\right)\right)\in\mathcal{C}_{q}\left(\mathbb{X}\right)$
be an arbitrary continuous MO function over $\mathbb{X}$. By Theorem
\ref{thm: Wang}, there exists an MO mean function
\[
m_{j}\left(\bm{x};\bm{\theta}_{j}\right)=\sum_{z=1}^{n_{j}}\frac{\pi_{jz}\phi_{p}\left(\bm{x};\bm{\mu}_{jz},\bm{\Sigma}_{jz}\right)}{\sum_{\zeta=1}^{n}\pi_{j\zeta}\phi_{p}\left(\bm{x};\bm{\mu}_{j\zeta},\bm{\Sigma}_{j\zeta}\right)}a_{jz}\text{,}
\]
for each $j\in\left[q\right]$, such that $\text{d}_{\infty}\left(m_{j}\left(\bm{x};\bm{\theta}_{j}\right),u_{j}\left(\bm{x}\right)\right)<\epsilon/q$,
for any $\epsilon>0$. Here, $\bm{\theta}_{j}$ is a parameter vector
that contains the unique elements of $\bm{\mu}_{jz}$, $\bm{\Sigma}_{jz}$,
and $a_{jz}\in\mathbb{R}$, for each $z\in n_{j}$, where $n_{j}\in\mathbb{N}$,
for each $j\in\left[q\right]$. Now, write

\[
\bm{m}_{j}\left(\bm{x};\bm{\theta}_{j}\right)=m_{j}\left(\bm{x};\bm{\theta}_{j}\right)\times\mathbf{e}_{j}
\]
and note that for any $k\ne j$, $\bm{m}_{jk}\left(\bm{x};\bm{\theta}_{j}\right)=0$,
for all $\bm{x}\in\mathbb{X}$.

Consider the fact that the $j\text{th}$ coordinate of the function
\begin{equation}
\bm{m}\left(\bm{x};\bm{\theta}\right)=\sum_{j=1}^{q}\bm{m}_{j}\left(\bm{x};\bm{\theta}_{j}\right)\label{eq: sum of yj}
\end{equation}
is only influenced by the $j\text{th}$ functional $\bm{m}_{j}\left(\bm{x};\bm{\theta}_{j}\right)$,
by construction. Thus, at each coordinate $j$, we have
\[
\text{d}_{\infty}\left(m_{j}\left(\bm{x}\right),u_{j}\left(\bm{x}\right)\right)=\text{d}_{\infty}\left(m_{j}\left(\bm{x};\bm{\theta}_{j}\right),u_{j}\left(\bm{x}\right)\right)<\epsilon/q\text{.}
\]
By definition of the induced distance, we therefore obtain the result
that
\begin{eqnarray*}
\text{d}_{q,\infty}\left(\bm{m},\bm{u}\right) & = & \sum_{j=1}^{q}\text{d}_{\infty}\left(m_{j}\left(\bm{x};\bm{\theta}_{j}\right),u_{j}\left(\bm{x}\right)\right)\\
 & < & q\times\left(\epsilon/q\right)=\epsilon\text{.}
\end{eqnarray*}

It suffices to show that (\ref{eq: sum of yj}) is a function in the
class $\mathcal{M}_{q}^{*}\left(\mathbb{X}\right)$. We obtain such
a result by repeated application of Lemma \ref{lem: closure under add}. 

\section{Discussion and conclusions}

Theorem \ref{thm: MO Norets} implies that all $q$ univariate $g_{Y_{j}|\bm{X}}\left(y_{j}|\bm{x}\right)$
conditional PDFs ($j\in\left[q\right]$) of a $q\text{-variate}$
target conditional PDF $g_{\bm{Y}|\bm{X}}\left(\bm{y}|\bm{x}\right)$
can be approximated to an arbitrary degree of accuracy via a Gaussian
gated MoLE model with Gaussian linear experts, with respect to a Kullback-Leibler
like divergence, assuming the fulfillment of Assumptions {[}B1{]}
and {[}B2{]}. Unfortunately, the statement of the theorem provides
no guarantees regarding the approximation accuracy of the dependence
structures between each of the $q$ univariate variables $y_{j}$,
conditioned on the observation $\bm{X}=\bm{x}$. Using Theorem \ref{thm: Univariate case Norets},
we cannot prove such a result using algebraic manipulations alone,
in the manner that has been used to prove Theorem \ref{thm: MO Norets}.
Proving that dependence structures can also be approximated to an
arbitrary degree of accuracy is a topic of ongoing research in the
literature. Such results may be sought via adaptations and extensions
of the joint density approximation results of \citet[Sec. 33.1]{DasGupta2008}
or \citet{Norets2012}, to the problem of multivariate conditional
density approximation.

Finally, we note that Theorems \ref{thm: Univariate case Norets}\textendash \ref{thm: Denseness of Mean}
do not provide rates, regarding the reduction of approximation error
as functions of $q$ and $n$. Rate results would require stronger
assumptions on the space of approximands. For example, we may utilize
the results of \citet{Zeevi1998} in order to obtain an approximation
rate for functional approximations from the class $\mathcal{M}_{q}\left(\mathbb{X}\right)$,
under the additional assumption that the MO approximand is a member
of some appropriate Sobolev space. Similarly, using the results of
\citet{Jiang1999a}, we may obtain approximation rates for conditional
approximations from the class $\mathcal{L}_{q}\left(\mathbb{X}\right)$,
under the additional assumptions that the approximand univariate conditional
PDFs satisfy are restricted to affine-dependence structures, with
respect to the input vector.

In this paper, we sought to prove the most general results that were
available, regarding the approximation capability of the Gaussian
gated MoLE model. As such, we do not wish to impose more assumptions
than is strictly necessary in order to establish meaningful theorems.
We leave the establishment of further interesting results that may
require more stringent assumptions to the future.

\section*{Acknowledgments}

Hien Nguyen is funded by Australian Research Council (ARC) grants
DE170101134 andDP180101192, and a La Trobe University startup grant.
This research is funded directly by the Inria LANDeR project. The
authors thank the anonymous Reviewers for their useful comments that
have improved the exposition of the article.

\section*{Appendix}

\subsection*{The induced norm}

Let $\mathcal{U}$ be a normed vector space, and let $\bm{u}$ and
$\bm{v}$ be arbitrary elements of $\mathcal{U}$. We say that the
operator $\left\Vert \cdot\right\Vert $ is a norm on $\mathcal{U}$
if it satisfies the following assumptions: (i) $\left\Vert \bm{u}\right\Vert \ge0$
and $\left\Vert \bm{u}\right\Vert =0$ if and only if $\bm{u}=\mathbf{0}$,
(ii) For every $c\in\mathbb{R}$, $\left\Vert c\bm{u}\right\Vert =\left|c\right|\left\Vert \bm{u}\right\Vert $,
and (iii) $\left\Vert \bm{u}+\bm{v}\right\Vert \le\left\Vert \bm{u}\right\Vert +\left\Vert \bm{v}\right\Vert $
(cf. \citealp[Sec. 4.6]{Oden2010}).
\begin{prop}
For any vector space $\mathcal{U}_{q}\left(\mathbb{X}\right)$ of
MO functions on $\mathbb{X}$, the operator $\left\Vert \cdot\right\Vert _{\Sigma}$
satisfies the definition of a norm.
\end{prop}
\begin{proof}
Let $\bm{u}^{\top}=\left(u_{1},\dots u_{q}\right)$ and $\bm{v}^{\top}=\left(v_{1},\dots,v_{q}\right)$
be two arbitrary elements in $\mathcal{U}_{q}$. Recall that the operator
$\left\Vert \cdot\right\Vert _{\infty}$ is a norm over any vector
space of single-output functions. This implies that $\left\Vert u_{j}\right\Vert _{\infty}\ge0$
for each $j\in\left[q\right]$ and thus $\left\Vert \bm{u}\right\Vert _{q,\infty}=\sum_{j=1}^{q}\left\Vert u_{j}\right\Vert _{\infty}\ge0$.

Suppose that $\left\Vert \bm{u}\right\Vert _{q,\infty}=0$. This implies
that each component of $\sum_{j=1}^{q}\left\Vert u_{j}\right\Vert _{\infty}$
must equal to zero since no component may take a negative value. However,
since $\left\Vert \cdot\right\Vert _{\infty}$ is a norm, this implies
that $\bm{u}=\mathbf{0}$. Now suppose that $\bm{u}=\mathbf{0}$.
The direct definition of $\left\Vert \cdot\right\Vert _{q,\infty}$
leads to the result that $\left\Vert \bm{u}\right\Vert _{q,\infty}=0$.
Thus, together, $\left\Vert \cdot\right\Vert _{q,\infty}$ fulfills
Assumption (i).

Assumption (ii) is shown to be fulfilled by observing the direct chain
of equalities:
\begin{align*}
\left\Vert c\bm{u}\right\Vert _{q,\infty} & =\sum_{j=1}^{q}\left\Vert cu_{j}\right\Vert _{\infty}\\
 & =\sum_{j=1}^{q}\left|c\right|\left\Vert u_{j}\right\Vert _{\infty}\\
 & =\left|c\right|\sum_{j=1}^{q}\left\Vert u_{j}\right\Vert _{\infty}\\
 & =\left|c\right|\left\Vert \bm{u}\right\Vert _{q,\infty}\text{,}
\end{align*}
where the second line is due to the fact that $\left\Vert \cdot\right\Vert _{\infty}$
is a norm.

Assumption (iii) is also shown to be fulfilled by observing the chain
of arguments:
\begin{align*}
\left\Vert \bm{u}+\bm{v}\right\Vert _{q,\infty} & =\sum_{j=1}^{q}\left\Vert u_{j}+v_{j}\right\Vert _{\infty}\\
 & \le\sum_{j=1}^{q}\left[\left\Vert u_{j}\right\Vert _{\infty}+\left\Vert v_{j}\right\Vert _{\infty}\right]\\
 & =\sum_{j=1}^{q}\left\Vert u_{j}\right\Vert _{\infty}+\sum_{j=1}^{q}\left\Vert v_{j}\right\Vert _{\infty}\\
 & =\left\Vert \bm{u}\right\Vert _{q,\infty}+\left\Vert \bm{v}\right\Vert _{q,\infty}\text{,}
\end{align*}
where the second line is again due to the fact that $\left\Vert \cdot\right\Vert _{\infty}$
is a norm. The proof is thus complete.
\end{proof}

\subsection*{Proof of Lemma \ref{lem: closure under add}}

Since $\bm{m}_{\left[1\right]},\bm{m}_{\left[2\right]}\in\mathcal{M}_{q}^{*}\left(\mathbb{X}\right)$,
we can write $\bm{y}_{\left[k\right]}$ as
\[
\bm{m}_{\left[k\right]}\left(\bm{x};\bm{\theta}_{k}\right)=\sum_{z=1}^{n_{k}}\frac{\pi_{kz}\phi_{p}\left(\bm{x};\bm{\mu}_{kz},\bm{\Sigma}_{kz}\right)}{\sum_{\zeta=1}^{n_{k}}\pi_{k\zeta}\phi_{p}\left(\bm{x};\bm{\mu}_{k\zeta},\bm{\Sigma}_{k\zeta}\right)}\bm{a}_{kz}\text{,}
\]
where $\bm{\theta}_{k}$ contains the unique elements of $\bm{\mu}_{kz}$,
$\bm{\Sigma}_{kz}$, and $\bm{a}_{kz}$ ($z\in\left[n_{k}\right]$;
$n_{k}\in\mathbb{N}$), for each $k\in\left\{ 1,2\right\} $. 

Next, we write
\begin{eqnarray*}
\bm{m}_{\left[12\right]}\left(\bm{x}\right) & = & \sum_{k=1}^{2}\sum_{z=1}^{n_{k}}\frac{\pi_{kz}\phi_{p}\left(\bm{x};\bm{\mu}_{kz},\bm{\Sigma}_{kz}\right)}{\sum_{\zeta=1}^{n_{k}}\pi_{k\zeta}\phi_{p}\left(\bm{x};\bm{\mu}_{k\zeta},\bm{\Sigma}_{k\zeta}\right)}\bm{a}_{kz}\\
 & = & \sum_{z=1}^{n_{1}}\frac{\pi_{1z}\phi_{p}\left(\bm{x};\bm{\mu}_{1z},\bm{\Sigma}_{1z}\right)\sum_{\zeta=1}^{n_{2}}\pi_{2\zeta}\phi_{p}\left(\bm{x};\bm{\mu}_{2\zeta},\bm{\Sigma}_{2\zeta}\right)}{\prod_{k=1}^{2}\sum_{\zeta=1}^{n_{k}}\pi_{k\zeta}\phi_{p}\left(\bm{x};\bm{\mu}_{k\zeta},\bm{\Sigma}_{k\zeta}\right)}\bm{a}_{1z}\\
 &  & +\sum_{z=1}^{n_{2}}\frac{\pi_{2z}\phi_{p}\left(\bm{x};\bm{\mu}_{2z},\bm{\Sigma}_{2z}\right)\sum_{\zeta=1}^{n_{1}}\pi_{1\zeta}\phi_{p}\left(\bm{x};\bm{\mu}_{1\zeta},\bm{\Sigma}_{1\zeta}\right)}{\prod_{k=1}^{2}\sum_{\zeta=1}^{n_{k}}\pi_{k\zeta}\phi_{p}\left(\bm{x};\bm{\mu}_{k\zeta},\bm{\Sigma}_{k\zeta}\right)}\bm{a}_{2z}\text{.}
\end{eqnarray*}
For each $s\in\left[n_{1}\right]$ and $t\in\left[n_{2}\right]$,
we can perform the following mappings: $\bm{a}_{\left(st\right)}=\bm{a}_{1s}+\bm{a}_{2t}$,
$\bar{\pi}_{\left(st\right)}=\pi_{1s}\pi_{2t}$, $\bm{\Sigma}_{\left(st\right)}^{-1}=\bm{\Sigma}_{1s}^{-1}+\bm{\Sigma}_{2t}^{-1}$,
and $\bm{\mu}_{\left(st\right)}=\bm{\Sigma}_{\left(st\right)}\left(\bm{\Sigma}_{1s}^{-1}\bm{\mu}_{1s}+\bm{\Sigma}_{2t}^{-1}\bm{\mu}_{2t}\right)$.

Using Lemma \ref{lem: product}, we can write
\begin{eqnarray*}
\bm{m}_{\left[12\right]}\left(\bm{x}\right) & = & \sum_{s=1}^{n_{1}}\sum_{t=1}^{n_{2}}\frac{c_{st}\bar{\pi}_{\left(st\right)}\phi_{d}\left(\bm{x};\bm{\mu}_{\left(st\right)},\bm{\Sigma}_{\left(st\right)}\right)}{\sum_{\xi=1}^{n_{1}}\sum_{\zeta=1}^{n_{2}}c_{\xi\zeta}\bar{\pi}_{\left(\xi\zeta\right)}\phi_{d}\left(\bm{x};\bm{\mu}_{\left(\xi\zeta\right)},\bm{\Sigma}_{\left(\xi\zeta\right)}\right)}\bm{a}_{\left(st\right)}\\
 & = & \sum_{s=1}^{n_{1}}\sum_{t=1}^{n_{2}}\frac{\pi_{\left(st\right)}\phi_{d}\left(\bm{x};\bm{\mu}_{\left(st\right)},\bm{\Sigma}_{\left(st\right)}\right)}{\sum_{\xi=1}^{n_{1}}\sum_{\zeta=1}^{n_{2}}\pi_{\left(\xi\zeta\right)}\phi_{d}\left(\bm{x};\bm{\mu}_{\left(\xi\zeta\right)},\bm{\Sigma}_{\left(\xi\zeta\right)}\right)}\bm{a}_{\left(st\right)}\text{,}
\end{eqnarray*}
where $\pi_{\left(st\right)}=c_{st}\bar{\pi}_{\left(st\right)}/\sum_{\xi=1}^{n_{1}}\sum_{\zeta=1}^{n_{2}}c_{\xi\zeta}\bar{\pi}_{\left(\xi\zeta\right)}$,
for each $s$ and $t$. Note that this implies that $\pi_{\left(st\right)}>0$
($s\in\left[n_{1}\right]$, $t\in\left[n_{2}\right]$) and $\sum_{s=1}^{n_{1}}\sum_{t=1}^{n_{2}}\pi_{\left(st\right)}=1$,
as required, since $c_{st}>0$.

Finally, utilizing some pairing function (see e.g., \citealp[Sec. 1.3]{Smorynski1991}),
we may map every pair $\left(s,t\right)\in\left[n_{1}\right]\times\left[n_{2}\right]$
uniquely to a $z\in\left[n_{\left[12\right]}\right]$, where $n_{\left[12\right]}=n_{1}n_{2}$.
Using this mapping, we can then write
\begin{eqnarray*}
\bm{m}_{\left[12\right]}\left(\bm{x}\right) & = & \sum_{z=1}^{n_{\left[12\right]}}\frac{\pi_{\left[12\right]z}\phi_{d}\left(\bm{x};\bm{\mu}_{\left[12\right]z},\bm{\Sigma}_{\left[12\right]z}\right)}{\sum_{\zeta=1}^{n_{\left[12\right]}}\pi_{\left[12\right]\zeta}\phi_{d}\left(\bm{x};\bm{\mu}_{\left[12\right]\zeta},\bm{\Sigma}_{\left[12\right]\zeta}\right)}\bm{a}_{\left[12\right]z}\\
 & = & \bm{m}_{\left[12\right]}\left(\bm{x};\bm{\theta}_{\left[12\right]}\right)\text{,}
\end{eqnarray*}
where $\bm{\theta}_{\left[12\right]}$ is a parameter vector that
contains the unique elements of $\pi_{\left[12\right]z}$, $\bm{\mu}_{\left[12\right]z}$,
$\bm{\Sigma}_{\left[12\right]z}$, and $\bm{a}_{\left[12\right]z}$
for each $z\in\left[n_{\left[12\right]}\right]$. Thus, we have shown
that $\bm{m}_{\left[12\right]}=\bm{m}_{\left[1\right]}+\bm{m}_{\left[2\right]}$
is in the class of functions $\mathcal{M}_{q}^{*}\left(\mathbb{X}\right)$.

\subsection*{Proof of Lemma \ref{lem: closure under mult}}

Since $f_{\left[1\right]}\in\mathcal{L}_{q}^{*}\left(\mathbb{X}\right)$
and $f_{\left[2\right]}\in\mathcal{L}_{r}^{*}\left(\mathbb{X}\right)$,
we can write
\[
f_{\left[1\right]}\left(\bm{y}_{\left[1\right]}|\bm{x};\bm{\theta}_{1}\right)=\sum_{z=1}^{n_{1}}\frac{\pi_{1z}\phi_{p}\left(\bm{x};\bm{\mu}_{1z},\bm{\Sigma}_{1z}\right)\phi_{q}\left(\bm{y}_{\left[1\right]};\bm{a}_{1z},\mathbf{C}_{1z}\right)}{\sum_{\zeta=1}^{n_{1}}\pi_{1\zeta}\phi_{p}\left(\bm{x};\bm{\mu}_{1\zeta},\bm{\Sigma}_{1\zeta}\right)}\text{,}
\]
and
\[
f_{\left[2\right]}\left(\bm{y}_{\left[2\right]}|\bm{x};\bm{\theta}_{2}\right)=\sum_{z=1}^{n_{2}}\frac{\pi_{2z}\phi_{p}\left(\bm{x};\bm{\mu}_{2z},\bm{\Sigma}_{2z}\right)\phi_{r}\left(\bm{y}_{\left[2\right]};\bm{a}_{2z},\mathbf{C}_{2z}\right)}{\sum_{\zeta=1}^{n_{2}}\pi_{2\zeta}\phi_{p}\left(\bm{x};\bm{\mu}_{2\zeta},\bm{\Sigma}_{2\zeta}\right)}\text{,}
\]
where $\bm{\theta}_{k}$ contains the unique elements of $\bm{\mu}_{kz}$,
$\bm{\Sigma}_{kz}$, $\bm{a}_{kz}$, and $\mathbf{C}_{kz}$ ($z\in\left[n_{k}\right]$;
$n_{k}\in\mathbb{N}$), for each $k\in\left\{ 1,2\right\} $. Here,
$\bm{y}^{\top}=\left(\bm{y}_{\left[1\right]}^{\top},\bm{y}_{\left[2\right]}^{\top}\right)$,
where $\bm{y}_{\left[1\right]}\in\mathbb{R}^{q}$ and $\bm{y}_{2}\in\mathbb{R}^{r}$.

Next, write
\[
f_{\left[12\right]}\left(\bm{y}|\bm{x}\right)=\sum_{z=1}^{n_{1}}\frac{\pi_{1z}\phi_{p}\left(\bm{x};\bm{\mu}_{1z},\bm{\Sigma}_{1z}\right)\phi_{q}\left(\bm{y}_{\left[1\right]};\bm{a}_{1z},\mathbf{C}_{1z}\right)}{\sum_{\zeta=1}^{n_{1}}\pi_{1\zeta}\phi_{p}\left(\bm{x};\bm{\mu}_{1\zeta},\bm{\Sigma}_{1\zeta}\right)}\times\sum_{z=1}^{n_{2}}\frac{\pi_{2z}\phi_{p}\left(\bm{x};\bm{\mu}_{2z},\bm{\Sigma}_{2z}\right)\phi_{r}\left(\bm{y}_{\left[2\right]};\bm{a}_{2z},\mathbf{C}_{2z}\right)}{\sum_{\zeta=1}^{n_{2}}\pi_{2\zeta}\phi_{p}\left(\bm{x};\bm{\mu}_{2\zeta},\bm{\Sigma}_{2\zeta}\right)}\text{,}
\]
and make the following mapping for each $s\in\left[n_{1}\right]$
and $t\in\left[n_{2}\right]$: $\bar{\pi}_{\left(st\right)}=\pi_{1s}\pi_{2t}$,
$\bm{\Sigma}_{\left(st\right)}^{-1}=\bm{\Sigma}_{1s}^{-1}+\bm{\Sigma}_{2t}^{-1}$,
and $\bm{\mu}_{\left(st\right)}=\bm{\Sigma}_{\left(st\right)}\left(\bm{\Sigma}_{1s}^{-1}\bm{\mu}_{1s}+\bm{\Sigma}_{2t}^{-1}\bm{\mu}_{2t}\right)$.
Furthermore, for each $s$ and $t$,
\begin{eqnarray*}
\phi_{q}\left(\bm{y}_{\left[1\right]};\bm{a}_{1s},\mathbf{C}_{1s}\right)\phi_{r}\left(\bm{y}_{\left[2\right]};\bm{a}_{2t},\mathbf{C}_{2t}\right) & = & \phi_{q+r}\left(\left[\begin{array}{c}
\bm{y}_{\left[1\right]}\\
\bm{y}_{\left[2\right]}
\end{array}\right];\left[\begin{array}{c}
\bm{a}_{1s}\\
\bm{a}_{2t}
\end{array}\right],\left[\begin{array}{cc}
\mathbf{C}_{1s} & \mathbf{0}\\
\mathbf{0} & \mathbf{C}_{2t}
\end{array}\right]\right)\\
 & = & \phi_{q+r}\left(\bm{y};\bm{a}_{\left(st\right)},\mathbf{C}_{\left(st\right)}\right)\text{,}
\end{eqnarray*}
specifies a $\left(q+r\right)\text{-dimensional}$ multivariate normal
PDF.

Using Lemma \ref{lem: product}, we can write
\begin{eqnarray*}
f_{\left[12\right]} & = & \sum_{s=1}^{n_{1}}\sum_{t=1}^{n_{2}}\frac{c_{st}\bar{\pi}_{\left(st\right)}\phi_{d}\left(\bm{x};\bm{\mu}_{\left(st\right)},\bm{\Sigma}_{\left(st\right)}\right)\phi_{q+r}\left(\bm{y};\bm{a}_{\left(st\right)},\mathbf{C}_{\left(st\right)}\right)}{\sum_{\xi=1}^{n_{1}}\sum_{\zeta=1}^{n_{2}}c_{\xi\zeta}\bar{\pi}_{\left(\xi\zeta\right)}\phi_{d}\left(\bm{x};\bm{\mu}_{\left(\xi\zeta\right)},\bm{\Sigma}_{\left(\xi\zeta\right)}\right)}\\
 & = & \sum_{s=1}^{n_{1}}\sum_{t=1}^{n_{2}}\frac{\pi_{\left(st\right)}\phi_{d}\left(\bm{x};\bm{\mu}_{\left(st\right)},\bm{\Sigma}_{\left(st\right)}\right)\phi_{q+r}\left(\bm{y};\bm{a}_{\left(st\right)},\mathbf{C}_{\left(st\right)}\right)}{\sum_{\xi=1}^{n_{1}}\sum_{\zeta=1}^{n_{2}}\pi_{\left(\xi\zeta\right)}\phi_{d}\left(\bm{x};\bm{\mu}_{\left(\xi\zeta\right)},\bm{\Sigma}_{\left(\xi\zeta\right)}\right)}\text{,}
\end{eqnarray*}
where $\pi_{\left(st\right)}=c_{st}\bar{\pi}_{\left(st\right)}/\sum_{\xi=1}^{n_{1}}\sum_{\zeta=1}^{n_{2}}c_{\xi\zeta}\bar{\pi}_{\left(\xi\zeta\right)}$,
for each $s$ and $t$. 

In a similar manner to the approach from Lemma \ref{lem: closure under add},
we may map every pair $\left(s,t\right)\in\left[n_{1}\right]\times\left[n_{2}\right]$
uniquely to a $z\in\left[n_{\left[12\right]}\right]$, where $n_{\left[12\right]}=n_{1}n_{2}$.
Using this mapping, we can then write
\begin{eqnarray*}
f_{\left[12\right]}\left(\bm{y}|\bm{x}\right) & = & \sum_{z=1}^{n_{\left[12\right]}}\frac{\pi_{\left[12\right]z}\phi_{d}\left(\bm{x};\bm{\mu}_{\left[12\right]z},\bm{\Sigma}_{\left[12\right]z}\right)\phi_{q+r}\left(\bm{y};\bm{a}_{\left(st\right)},\mathbf{C}_{\left(st\right)}\right)}{\sum_{\zeta=1}^{n_{\left[12\right]}}\pi_{\left[12\right]\zeta}\phi_{d}\left(\bm{x};\bm{\mu}_{\left[12\right]\zeta},\bm{\Sigma}_{\left[12\right]\zeta}\right)}\\
 & = & f_{\left[12\right]}\left(\bm{y}|\bm{x};\bm{\theta}_{\left[12\right]}\right)\text{,}
\end{eqnarray*}
where $\bm{\theta}_{\left[12\right]}$ is a parameter vector that
contains the unique elements of $\pi_{\left[12\right]z}$, $\bm{\mu}_{\left[12\right]z}$,
$\bm{\Sigma}_{\left[12\right]z}$, $\bm{a}_{\left[12\right]z}$, and
$\mathbf{C}_{\left[12\right]z}$, for each $z\in\left[n_{\left[12\right]}\right]$.
Thus, we have shown that $f_{\left[12\right]}=f_{\left[1\right]}f_{\left[2\right]}$
is in the class of functions $\mathcal{L}_{q+r}^{*}\left(\mathbb{X}\right)$.

\bibliographystyle{apalike2}
\bibliography{MASTERBIB}

\begin{thebibliography}{}

\bibitem[Azzalini \& {Dalla Valle}, 1996]{Azzalini1996}
Azzalini, A. \& {Dalla Valle}, A. (1996).
\newblock {the multivariate skew-normal distribution}.
\newblock {\em Biometrika}, 83, 715--726.

\bibitem[Bromiley, 2014]{Bromiley2014}
Bromiley, P.~A. (2014).
\newblock {\em {Products and convolutions of Gaussian probability density
  functions}}.
\newblock Technical Report 2003-003, TINA-VISION, Manchester.

\bibitem[Chamroukhi, 2016]{Chamroukhi2016}
Chamroukhi, F. (2016).
\newblock {Robust mixture of experts modeling using the t distribution}.
\newblock {\em Neural Networks}, 79, 20--36.

\bibitem[Chamroukhi, 2017]{Chamroukhi2017}
Chamroukhi, F. (2017).
\newblock {Skew t mixture of experts}.
\newblock {\em Neurocomputing}, 266, 390--408.

\bibitem[Chamroukhi et~al., 2013]{Chamroukhi2013a}
Chamroukhi, F., Mohammed, S., Trabelsi, D., Oukhellou, L., \& Amirat, Y.
  (2013).
\newblock {Joint segmentation of multivariate time series with hidden process
  regression for human activity recognition}.
\newblock {\em Neurocomputing}, 120, 633--644.

\bibitem[Chen et~al., 1999]{Chen1999}
Chen, K., Xu, L., \& Chi, H. (1999).
\newblock {Improved learning algorithms for mixture of experts in multiclass
  classification}.
\newblock {\em Neural Networks}, 12, 1229--1252.

\bibitem[Cheney \& Light, 2000]{CheneyLight2000}
Cheney, W. \& Light, W. (2000).
\newblock {\em {A Course in Approximation Theory}}.
\newblock Pacific Grove: Brooks/Cole.

\bibitem[Chiou et~al., 2014]{Chiou2014}
Chiou, J.-M., Chen, Y.-T., \& Yang, Y.-F. (2014).
\newblock {Multivariate functional principal component analysis: a
  normalization approach}.
\newblock {\em Statistica Sinica}, 24, 1571--1596.

\bibitem[DasGupta, 2008]{DasGupta2008}
DasGupta, A. (2008).
\newblock {\em Asymptotic Theory Of Statistics And Probability}.
\newblock New York: Springer.

\bibitem[Deleforge et~al., 2015a]{Deleforge2015a}
Deleforge, A., Forbes, F., \& Horaud, R. (2015a).
\newblock {Acoustic space learning for sound-source separation and localization
  on binaural manifolds}.
\newblock {\em International Journal of Neural Systems}, 25, 1440003.

\bibitem[Deleforge et~al., 2015b]{Deleforge2015}
Deleforge, A., Forbes, F., \& Horaud, R. (2015b).
\newblock {High-dimensional regression with Gaussian mixtures and
  partially-latent response variables}.
\newblock {\em Statistics and Computing}, 25, 893--911.

\bibitem[Fu et~al., 2018]{Fu:2018aa}
Fu, H., Gong, M., Wang, C., \& Tao, D. (2018).
\newblock {MoE-SPNet: a mixture of experts scene parsing network}.
\newblock {\em Pattern Recognition}, 84, 226--236.

\bibitem[Geweke \& Keane, 2007]{Geweke2007a}
Geweke, J. \& Keane, M. (2007).
\newblock {Smoothly mixing regressions}.
\newblock {\em Journal of Econometrics}, 138, 252--290.

\bibitem[Grun et~al., 2012]{Grun2012}
Grun, B., Kosmidis, I., \& Zeileis, A. (2012).
\newblock {Extended beta regression in R: shaken, stirred, mixed, and
  partitioned}.
\newblock {\em Journal of Statistical Software}, 48, 1--25.

\bibitem[Grun \& Leisch, 2008]{Grun2008}
Grun, B. \& Leisch, F. (2008).
\newblock {Flexmix version 2: finite mixtures with concomitant variables and
  varying and constant parameters}.
\newblock {\em Journal of Statistical Software}, 28, 1--35.

\bibitem[Ingrassia et~al., 2014]{Ingrassia2014}
Ingrassia, S., Minotti, S.~C., \& Punzo, A. (2014).
\newblock {Model-based clustering via linear cluster-weighted models}.
\newblock {\em Computational Statistics and Data Analysis}, 71, 159--182.

\bibitem[Ingrassia et~al., 2012]{Ingrassia2012}
Ingrassia, S., Minotti, S.~C., \& Vittadini, G. (2012).
\newblock {Local statistical modeling via a cluster-weighted approach with
  elliptical distributions}.
\newblock {\em Journal of Classification}, 29, 363--401.

\bibitem[Jacobs et~al., 1991]{Jacobs1991}
Jacobs, R.~A., Jordan, M.~I., Nowlan, S.~J., \& Hinton, G.~E. (1991).
\newblock {Adaptive mixtures of local experts}.
\newblock {\em Neural Computation}, 3, 79--87.

\bibitem[Jiang \& Tanner, 1999a]{Jiang1999a}
Jiang, W. \& Tanner, M.~A. (1999a).
\newblock {Hierachical mixtures-of-experts for exponential family regression
  models: approximation and maximum likelihood estimation}.
\newblock {\em Annals of Statistics}, 27, 987--1011.

\bibitem[Jiang \& Tanner, 1999b]{Jiang1999}
Jiang, W. \& Tanner, M.~A. (1999b).
\newblock {On the approximation rate of hierachical mixtures-of-experts for
  generalized linear models}.
\newblock {\em Neural Computation}, 11, 1183--1198.

\bibitem[Jordan \& Jacobs, 1994]{Jordan1994}
Jordan, M.~I. \& Jacobs, R.~A. (1994).
\newblock {Hierachical mixtures of experts and the EM algorithm}.
\newblock {\em Neural Computation}, 6, 181--214.

\bibitem[Jordan \& Xu, 1995]{Jordan1995}
Jordan, M.~I. \& Xu, L. (1995).
\newblock {Convergence results for the EM approach to mixtures of experts
  architectures}.
\newblock {\em Neural Networks}, 8, 1409--1431.

\bibitem[Kalliovirta et~al., 2016]{Kalliovirta2016}
Kalliovirta, L., Meitz, M., \& Saikkonen, P. (2016).
\newblock {Gaussian mixture vector autoregression}.
\newblock {\em Journal of Econometrics}, 192, 485--498.

\bibitem[Masoudnia \& Ebrahimpour, 2014]{Masoudnia2014}
Masoudnia, S. \& Ebrahimpour, R. (2014).
\newblock {Mixture of experts: a literature survey}.
\newblock {\em Artificial Intelligence Review}, 42, 275--293.

\bibitem[Mendes \& Jiang, 2012]{MendesJiang2012}
Mendes, E.~F. \& Jiang, W. (2012).
\newblock {On convergence rates of mixture of polynomial experts}.
\newblock {\em Neural Computation}, 24.
\newblock 3025-3051.

\bibitem[Montuelle \& {Le Pennec}, 2014]{Montuelle2014}
Montuelle, L. \& {Le Pennec}, E. (2014).
\newblock {Mixture of Gaussian regressions model with logistic weights, a
  penalized maximum likelihood approach}.
\newblock {\em Electronic Journal of Statistics}, 8, 1661--1695.

\bibitem[Nguyen \& Chamroukhi, 2018]{Nguyen2018}
Nguyen, H.~D. \& Chamroukhi, F. (2018).
\newblock {Practical and theoretical aspects of mixture-of-experts modeling: an
  overview}.
\newblock {\em WIREs Data Mining and Knowledge Discovery}, (pp.\ e1246).

\bibitem[Nguyen et~al., 2016]{Nguyen2016a}
Nguyen, H.~D., {Lloyd-Jones}, L.~R., \& McLachlan, G.~J. (2016).
\newblock {A universal approximation theorem for mixture-of-experts models}.
\newblock {\em Neural Computation}, 28, 2585--2593.

\bibitem[Nguyen \& McLachlan, 2016]{Nguyen2016}
Nguyen, H.~D. \& McLachlan, G.~J. (2016).
\newblock {Laplace mixture of linear experts}.
\newblock {\em Computational Statistics and Data Analysis}, 93, 177--191.

\bibitem[Norets, 2010]{Norets2010}
Norets, A. (2010).
\newblock {Approximation of conditional densities by smooth mixtures of
  regressions}.
\newblock {\em Annals of Statistics}, 38, 1733--1766.

\bibitem[Norets \& Pati, 2017]{Norets2017}
Norets, A. \& Pati, D. (2017).
\newblock {Adaptive Bayesian estimation of conditional densities}.
\newblock {\em Econometric Theory}, 33, 980--1012.

\bibitem[Norets \& Pelenis, 2012]{Norets2012}
Norets, A. \& Pelenis, J. (2012).
\newblock {Bayesian modeling of joint and conditional distributions}.
\newblock {\em Journal of Econometrics}, 168, 332--346.

\bibitem[Norets \& Pelenis, 2014]{Norets2014}
Norets, A. \& Pelenis, J. (2014).
\newblock {Posterior consistency in conditional density estimation by covariate
  dependent mixtures}.
\newblock {\em Econometric Theory}, 30, 606--646.

\bibitem[Oden \& Demkowicz, 2010]{Oden2010}
Oden, J.~T. \& Demkowicz, L.~F. (2010).
\newblock {\em {Applied Functional Analysis}}.
\newblock Boca Raton: CRC Press.

\bibitem[Pelenis, 2014]{Pelenis2014}
Pelenis, J. (2014).
\newblock {Bayesian regression with heteroscedastic error density and
  parametric mean function}.
\newblock {\em Journal of Econometrics}, 178, 624--638.

\bibitem[Perthame et~al., 2018]{Perthame2018}
Perthame, E., Forbes, F., \& Deleforge, A. (2018).
\newblock {Inverse regression approach to robust nonlinear high-to-low
  dimensional mapping}.
\newblock {\em Journal of Multivariate Analysis}, 163, 1--14.

\bibitem[Pinkus, 2015]{Pinkus2015}
Pinkus, A. (2015).
\newblock {\em {Ridge Functions}}.
\newblock Cambridge: Cambridge University Press.

\bibitem[Pollard, 2002]{Pollard2002}
Pollard, D. (2002).
\newblock {\em {A User's Guide to Measure Theoretic Probability}}.
\newblock Cambridge: Cambridge University Press.

\bibitem[Prado et~al., 2006]{Prado2006}
Prado, R., Molina, F., \& Huerta, G. (2006).
\newblock {Multivariate time series modeling and classification via hierachical
  VAR mixture}.
\newblock {\em Computational Statistics and Data Analysis}, 51, 1445--1462.

\bibitem[Shazeer et~al., 2017]{Shazeer:2017aa}
Shazeer, N., Mirhoseini, A., Maziarz, K., Davis, A., Le, Q., Hinton, G., \&
  Dean, J. (2017).
\newblock {Outrageously large neural networks: the sparsely-gated
  mixture-of-experts layer}.
\newblock In {\em Proceedings of the International Conference on Learning
  Representation}.

\bibitem[Smorynski, 1991]{Smorynski1991}
Smorynski, C. (1991).
\newblock {\em {Logical Number Theory I: An Introduction}}.
\newblock Berlin: Springer.

\bibitem[Stone, 1948]{Stone1948}
Stone, M.~H. (1948).
\newblock {The generalized Weierstrass approximation theorem}.
\newblock {\em Mathematical Magazine}, 21, 237--254.

\bibitem[Wang \& Mendel, 1992]{Wang1992}
Wang, L.-X. \& Mendel, J.~M. (1992).
\newblock {Fuzzy basis functions, universal approximation, and orthogonal
  least-squares learning}.
\newblock {\em IEEE Transactions on Neural Networks}, 3, 807--814.

\bibitem[Xu et~al., 1995]{Xu1995}
Xu, L., Jordan, M.~I., \& Hinton, G.~E. (1995).
\newblock {An alternative model for mixtures of experts}.
\newblock In {\em Advances in Neural Information Processing Systems}  (pp.\
  633--640).

\bibitem[Yuksel et~al., 2012]{Yuksel2012}
Yuksel, S.~E., Wilson, J.~N., \& Gader, P.~D. (2012).
\newblock {Twenty years of mixture of experts}.
\newblock {\em IEEE Transactions on Neural Networks and Learning Systems}, 23,
  1177--1193.

\bibitem[Zeevi et~al., 1998]{Zeevi1998}
Zeevi, A.~J., Meir, R., \& Maiorov, V. (1998).
\newblock {Error bounds for functional approximation and estimation using
  mixtures of experts}.
\newblock {\em IEEE Transactions on Information Theory}, 44, 1010--1025.

\bibitem[Zhao et~al., 2018]{Zhao:2018aa}
Zhao, T., Chen, Q., Kuang, Z., Yu, J., Zhang, W., \& Fan, J. (2018).
\newblock {Deep mixture of diverse experts for large-scale visual recognition}.
\newblock {\em IEEE Transactions on Pattern Analysis and Machine Intelligence},
  41, 1072--1087.

\end{thebibliography}

\end{document}